\documentclass[10pt,aps,jmp,superscriptaddress]{revtex4-1}
\usepackage{amsfonts}

\usepackage{amsmath}
\usepackage{graphicx,color}
\usepackage{epstopdf}
\usepackage{epsfig}
\usepackage{mathrsfs}
\usepackage[breaklinks=true]{hyperref}

\newcommand{\beq}{\begin{equation}}
\newcommand{\eeq}{\end{equation}}
\newcommand{\bqa}{\begin{eqnarray}}
\newcommand{\eqa}{\end{eqnarray}}

\definecolor{green}{rgb}{0.00,0.50,0.00}

\newtheorem{theorem}{Theorem}

\newtheorem{proposition}[theorem]{Proposition}
\newtheorem{remark}[theorem]{Remark}

\newenvironment{proof}[1][Proof]{\noindent\textbf{#1.} }{\ \rule{0.5em}{0.5em}}

\begin{document}

\title{The Gisin-Percival Stochastic Schr\"{o}dinger Equation from
Standard Quantum Filtering Theory}

\author{John E.~Gough} \email{jug@aber.ac.uk}
   \affiliation{Aberystwyth University, SY23 3BZ, Wales, United Kingdom}
\date{\today}

\begin{abstract}
We show that the quantum state diffusion equation of Gisin and Percival, driven by complex Wiener noise, is equivalent up to a global stochastic phase to quantum trajectory models.  With an appropriate feedback scheme, we set up an analogue continuous measurement model with exactly simulates the Gisin-Percival quantum state diffusion.
\end{abstract}

\maketitle

\affiliation{Institute for Mathematics, Physics, and Computer Sciences, Aberystwyth University,
SY23 3BZ, Wales, United Kingdom}
 
\section{Introduction}
Standard quantum mechanics tells us that closed systems evolve according to the Schr\"{o}dinger equation, which provides a deterministic and reversible rule for states. The idea that quantum systems may evolve as stochastic processes has emerged in several ways: as an extension of standard formalism to allow for continuous measurements, as a numerical technique for simulating open systems, and as a new law of Nature. These approaches have very similar requirements and this leads to a convergence of structure that is pleasing from a mathematical point of view, \cite{Decoherence}. However, this is commonality implies that distinguishing between different physical situations is difficult. 
 
The Gisin-Percival model \cite{Gisin-Percival} is exceptional in that it is postulated on the grounds of covariance (up to a global phase) under transformations of the coupling operators that leave the GKS-Lindblad generator invariant (a symmetry that is typically broken in quantum filtering/trajectories models), and it achieves this by using specific complex Wiener possesses as driving noise instead of the more typical real Brownian motions. 

This leads to the question of whether it is possible to realize the Gisin-Percival equation as a special instance of a quantum trajectory model, or whether it is truly distinguished from this class. 
The problem was first raised and addressed by Wiseman and Milburn \cite{WM93}: they consider a heterodyne detection scheme with a local oscillator with a finite detuning $\Omega$ from the the system
and show that Gisin equation follows as the infinite detuning limit of the associated stochastic Schr\"{o}dinger equation, see also \cite{WM09} - the particular structure of the Gisin-Percival equation arising in effect from the vanishing of the rapidly oscillating terms.

We treat this problem directly in Section \ref{sec:Rep} without recourse to limits. In Proposition \ref{prop:bv2}) we show that an explicit homodyne quantum trajectory problem can reproduce the Gisin-Percival solution up to a global stochastic phase that depends causally on the noise. Though, in Proposition \ref{prop:bv}, we show that the stochastic process must have nontrivial quadratic variation. In Section \ref{sec:simulation}, we show that it actually possible, with the additional use of feedback, to obtain a standard homodyne detection scheme that exactly reproduces the Gisin-Percival evolution.

\subsection{Stochastic Schr\"{o}dinger Equations}

We fix a Hilbert space $\mathfrak{h}$ as the system state space, and all operators will be assumed to act on this space. 
Let $\mathbf{L}=[L_{1},L_{2},\cdots ,L_{n}]^{\top }$ be a collection of operators, and $H$ be a self-adjoint operator. 
The $\left( \mathbf{L},H\right) $ \textit{Belavkin-Schr\"{o}dinger equation} takes the form \cite{Decoherence}
\begin{eqnarray}
d|\psi _{t}\rangle  =\left( -iH-\frac{1}{2}\sum_{k}\left( L_{k}^{\ast
}L_{k}-\lambda _{k}L_{k}+\frac{1}{4}\lambda _{k}^{2}\right) \right) |\psi
_{t}\rangle \,dt  +\sum_{k}\left( L_{k}-\frac{1}{2}\lambda _{k}\right) |\psi _{t}\rangle
\,dI_{k}\left( t\right)   \label{eq:BS_equation}
\end{eqnarray}
where 
\begin{eqnarray}
\lambda _{k}\left( t\right) =\langle \psi _{t}|\left( L_{k}+L_{k}^{\ast
}\right) \psi _{t}\rangle ,
\label{eq:lambda}
\end{eqnarray}
and $\left\{ I_{k}\right\} $ are a family of independent standard Wiener processes: 
\begin{eqnarray}
dI_{j}\left( t\right) \,dI_{k}\left( t\right) =\delta _{jk}dt.
\label{eq:I}
\end{eqnarray}

In contrast, let $\mathbf{R}=[R_{1},R_{2},\cdots ,R_{m}]^{\top }$ be a collection of operators, then the 
$\left( \mathbf{R},H\right) $ \textit{Gisin-Percival-Schr\"{o}dinger equation} takes the form \cite{Decoherence}
\begin{eqnarray}
d|\tilde{\psi}_{t}\rangle  =\left( -iH-\frac{1}{2}\sum_{k}\left(
R_{k}^{\ast }R_{k}-2c_{k}^{\ast }R_{k}+|c_{k}|^{2}\right) \right) |\tilde{%
\psi}_{t}\rangle \,dt +\sum_{k}\left( R_{k}-c_{k}\right) |\tilde{\psi}_{t}\rangle \,d\xi
_{k}\left( t\right) ^{\ast }  \label{eq:GS_equation}
\end{eqnarray}
where now
\begin{eqnarray}
c_{k}\left( t\right) =\langle \psi _{t}|L_{k}\,\psi _{t}\rangle ,
\label{eq:c}
\end{eqnarray}
and $\left\{ \xi _{k}\right\} $ are a family of independent \emph{complex} Wiener
processes: 
\begin{eqnarray}
d\xi _{j}\left( t\right) ^{\ast }\,d\xi _{k}\left( t\right)  &=&\delta
_{jk}dt, \notag \\
d\xi _{j}\left( t\right) \,d\xi _{k}\left( t\right)  &=&0.
\label{eq:xi}
\end{eqnarray}
The equations (\ref{eq:BS_equation}) and (\ref{eq:GS_equation}) are both stochastic analogues of the Schr\"{o}dinger equation, driven by classical noise. (Note that every complex Wiener process $\xi$ can be written as a sum, $ \frac{1}{\sqrt{2}} \big( B_1 (t)
+ i B_2(t) \big)$, of two independent standard Wiener processes. Both equations are nonlinear - on account of $\lambda(t)$'s in
(\ref{eq:BS_equation}) and the $c(t)$'s in (\ref{eq:GS_equation}).

The equations however have very different origins. The Belavkin-Schr\"{o}dinger equation
comes from conditioning the state of the system on the results of continuous indirect measurement \cite{Belavkin}, and is
best understood in terms of a quantum filtering theory \cite{BouGutMaa04}-\cite{BvHJ}, based on quantum stochastic calculus
\cite{HP84}-\cite{Par92}. It has also arisen independently as a quantum Monte Carlo technique \cite{CarBook93}-\cite{DCM}.
The Gisin-Percival-Schr\"{o}dinger equation however is postulated as a new law of Physics: it is proposed as the equation describing a quantum state diffusion leading to a collapse of the wavefunction even for \textit{closed} systems. It should be mentioned however, that (\ref{eq:BS_equation}) also occurs as a proposal for spontaneous collapse and localization \cite{GPR90}, again as a new
law of Physics.

The question we ask in this article is whether or not one can realize a particular  Gisin-Percival-Schr\"{o}dinger equation as some Belavkin-Schr\"{o}dinger equation. What we show is that for any fixed $\left( \mathbf{R},H\right) $, leading to a solution $\tilde{\psi} _t$ of the corresponding Gisin-Percival-Schr\"{o}dinger equation with initial state $\psi_0$, there will exist an $\left( \mathbf{L},H\right) $ such that (with appropriate identification of the noises)
\begin{eqnarray}
\vert \psi_t \rangle = e^{i \Theta (t)} \, \vert \tilde{\psi}_t \rangle 
\label{eq:main_result}
\end{eqnarray}
where $\psi _t$ is to be the solution of some $\left( \mathbf{L},H\right) $ Belavkin-Schr\"{o}dinger equation with the same
initial state $\psi_0$.

In particular, the wavefunctions are equivalent up to some c-number phase, $\Theta (t)$. However, this phase will have to be time-dependent. In fact, $\Theta (t)$ will have to be random, with a causal dependence on the noise up to time $t$, and in particular it must be a diffusion process with unbounded variation, that is
\begin{eqnarray}
d \Theta (t) \, d \Theta (t)  \neq 0.
\end{eqnarray}
In fact, we shall show in Proposition \ref{prop:bv} below that if $\Theta (t)$ is of bounded variation then (\ref{eq:main_result}) cannot hold.

In Section \ref{sec:representation}, we shall give explicit constructions to establish the representation (\ref{eq:main_result}).
After a discussion of the underlying network theory in Section \ref{sec:system}, we show in Section \ref{sec:simulation}that the Gisin-Percival Schr\"{o}dinger equation can in fact be modelled exactly as a Belavkin-Schr\"{o}dinger equation if we allow feedback.

\subsection{Relation to Other Work}
It is well known that the Belavkin-Schr\"{o}dinger equation has the linear form
\begin{eqnarray}
d \vert \chi_t \rangle = - \big( \frac{1}{2} \sum_k L^\ast_k L_k +i H \big)  \vert \chi_t \rangle \, dt 
+\sum_k L_k \vert \chi_t \rangle \, dY_k (t),
\label{eq:BZ}
\end{eqnarray}
where $Y_k (t)$ are the measured quadrature processes (usually from a), \cite{Belavkin}-\cite{BvHJ}.
The Wiener processes $I_k (t)$ are then the innovations processes given by
\begin{eqnarray}
dI_k (t) = dY_k (t) - \lambda_k (t)\, dt .
\end{eqnarray}
Underlying this, of course, is a quantum stochastic description for a definite input-system-output model. Effectively, one is homodyning the quantum output fields, with the $k$th channel interpreted as a continuous measurement of the observable $L_k +L^\ast_k$.
The linearity is no great surprise, and is the well known classically as the Zakai equation in filtering theory.

The fact that the $Y_k$ are \textit{not} Wiener processes (otherwise we would be conditioning on white noise) is often missed in the physics literature. One can consider the mathematical equation (\ref{eq:BZ}) with the $Y_k$ replaced by independent Wiener processes $Z_k$, but this is physically meaningless: it becomes physically correct only when one replaces the $Z_k$ with the processes having the same statistical distribution as the $Y_k$ and mathematically this can be done using a re-weighting of the path probabilities known as a Cameron-Martin-Girsanov transformation, see for instance \cite{Gatarek_Gisin}.

\subsubsection{Wiseman Milburn Limit}
The derivation by Wiseman and Milburn \cite{WM93} is based on an unequivocally clear physical physical model. The consider a system with single input field with coupling operator
\begin{eqnarray}
L = \sqrt{\gamma} e^{-i \varphi } e^{i \Omega t} \, a,
\end{eqnarray}
where $a$ is a fixed system operator (say the annihilator for a cavity mode) $\varphi$ a fixed phase and $\Omega $ is the detuning frequency from the local oscillator. Substituting into (\ref{eq:BZ}) we have
\begin{eqnarray*}
d \vert \chi_t \rangle = - \big( \frac{1}{2}  \gamma a^\ast a +i H \big)  \vert \chi_t \rangle \, dt + \sqrt{\gamma} e^{-i \varphi } e^{i \Omega t} \, a \vert \chi_t \rangle  \bigg( dI (t) + \sqrt{\gamma} \langle e^{-i \varphi } e^{i \Omega t} \, a + e^{i \varphi } e^{-i \Omega t} \, a^\ast \rangle dt\bigg) .
\label{eq:BZLO}
\end{eqnarray*}
They now take the limit $\Omega \to \infty$ and drop the rapidly oscillating terms to obtain
\begin{eqnarray}
d \vert \chi_t \rangle  \approx  - \big( \frac{1}{2}  \gamma a^\ast a +i H \big)  \vert \chi_t \rangle \, dt  + \sqrt{\gamma} a \vert \chi_t \rangle   \bigg( d\xi_\Omega (t)^\ast + \sqrt{\gamma} \langle  a^\ast \rangle dt\bigg) ,
\label{eq:BZLO_approx}
\end{eqnarray}
where they introduce the process $\xi_\Omega (t)$ defined by
\begin{eqnarray}
d\xi_\Omega (t)^\ast = e^{-i \varphi } e^{i \Omega t} \, dI(t).
\end{eqnarray}
One may argue that $ d\xi_\Omega (t)^\ast d\xi_\Omega (t) = \big( dI(t) \big)^2 =dt$ while $d\xi_\Omega (t)^2 = e^{2i \varphi } e^{-2i \Omega t} \, dt\approx 0$, so that in (\ref{eq:BZLO_approx})
we may replace $\xi_\Omega (t)$ with a limit complex Wiener process $\xi (t)$. The resulting equation then leads to the $(R=\sqrt{\gamma} a, H)$ Gisin-Percival equation for $\vert \chi_t \rangle / \| \chi_t \|$.
A precise derivation would presumably involve a Fourier analysis of the input (pre-measurement) process $Z(t) =
e^{-i \varphi } e^{i \Omega t}B(t) e^{i \varphi } e^{-i \Omega t} B(t)^\ast$ considered in \cite{BvH_ref}, and the use of the Riemann-Lebesgue Lemma \cite{Bochner} to justify the omission of rapidly oscillating terms, but this lies outside the scope of this paper.

\subsubsection{Girsanov Transformation}
We mention that a recent result by Parthasarathy and Usha Devi \cite{Partha_Usha} shows how to derive the Gisin-Percival equation from a quantum stochastic evolution using an appropriate Girsanov transformation. In fact, their scheme uses the same construction as that appearing in our simplest representation of the Gisin-Percival equations in terms of the Belavkin and the , see Subsection \ref{subsec:example}.

\subsection{Covariance}
\label{subec:Cov}
The equations (\ref{eq:BS_equation}) and (\ref{eq:GS_equation}) give rise to quantum dynamical semigroups: for instance,
taking the average over the ensemble of Wiener paths $\{ I_k \}$, we obtain
\begin{eqnarray*}
\mathbb{E} [ \langle \psi_t \vert X \, \psi_t \rangle ] \equiv \Phi_t (X)
\end{eqnarray*}
where $\Phi_t$ is the completely positive semigroup on the operators of $\mathfrak{h}$ with the GKS-Lindbladian associated with
 $  \left(  \mathbf{L},H\right)$ to be 
\begin{eqnarray}
\mathcal{L}_{\left(  \mathbf{L},H\right)}\left( X\right) =\frac{1}{2}\sum_{k}\bigg( \left[ L_{k}^{\ast
},\cdot \right] L_{k}+ L_{k}^{\ast }\left[ \cdot ,L_{k}%
\right] \bigg)-i\left[ \cdot ,H\right] . \notag \\
\end{eqnarray}
The corresponding expression involving $\tilde{\psi}_t$ averaged over the complex Wiener noise leads to the quantum dynamical semigroup with GKS-Lindbladian $  \left(  \mathbf{R},H\right)$.

Nevertheless, there are several differences. The GKS-Lindbladian $  \left(  \mathbf{L},H\right)$ is invariant under the transformations \cite{Decoherence}
\begin{eqnarray}
L_{k}\rightarrow L_{k}^{\prime }=\sum_{j}u_{kj}L_{k}  \label{eq:EU_rot}
\end{eqnarray}
with $U=\left[ u_{jk}\right] $ unitary, and the transformations 
\begin{eqnarray}
L_{k} &\rightarrow &L_{k}^{\prime }=L_{k}+\beta _{k},  \notag \\
H &\rightarrow &H^{\prime }=H+\sum_{k}\text{Im}\left\{ \beta _{k}^{\ast
}L_{k}\right\} +\epsilon , \label{eq:EU_transl}
\end{eqnarray}
with $\epsilon $ real.

The Gisin-Percival-Schr\"{o}dinger equation transforms covariantly under these transformations. For the unitary rotations
transformations (\ref{eq:EU_rot}) we need only rotate the complex Wiener processes, while under the translation transformations (\ref{eq:EU_transl}) we have 
\begin{eqnarray}
d|\tilde{\psi}_{t}^{\prime }\rangle =\left( -iH-\frac{1}{2}%
\sum_{k}\left( R_{k}^{\ast }R_{k}-2c_{k}^{\ast }R_{k}+|c_{k}|^{2}\right)
\right) |\tilde{\psi}_{t}^{\prime }\rangle \,dt 
-i\epsilon (t) |\tilde{\psi}_{t}^{\prime }\rangle \,dt   +\sum_{k}\left( R_{k}-c_{k}\right) |\tilde{\psi}_{t}^{\prime }\rangle
\,d\xi _{k}\left( t\right) ^{\ast },  \label{eq:GS_prime}
\end{eqnarray}
where the additional term $\epsilon (t)=\sum_{k}\text{Im}\left\{ \beta
_{k}^{\ast }c_{k}\right\} $ is a time-dependent phase.

The Belavkin-Schr\"{o}dinger equation, however, does not typically possess such a covariance.

\subsection{Notation}

We will write the Belavkin-Schr\"{o}dinger equation (\ref{eq:BS_equation}) in the form
\begin{eqnarray}
d|\psi _{t}\rangle  = dF_{ (\mathbf{L}, H)} (t) \, |\psi _{t}\rangle ,
\end{eqnarray}
with
\begin{eqnarray}
dF_{(\mathbf{L}, H)} (t)  =\sum_{k}\left( L_{k}-\frac{1}{2}\lambda _{k}\right)  
\,dI_{k}\left( t\right)+\left( -iH-\frac{1}{2}\sum_{k}\left( L_{k}^{\ast
}L_{k}-\lambda _{k}L_{k}+\frac{1}{4}\lambda _{k}^{2}\right) \right)   dt 
 .
\label{eq:F}
\end{eqnarray}

Similarly, we will write the Gisin-Percival-Schr\"{o}dinger equation (\ref{eq:GS_equation}) in the form
\begin{eqnarray}
d| \tilde{\psi }_{t}\rangle  = dM_{ (\mathbf{L}, H)} (t) \, |\tilde{\psi} _{t}\rangle ,
\end{eqnarray}
with
\begin{eqnarray}
 dM_{(\mathbf{R}, H)} (t)  =\sum_{k}\left( R_{k}-c_{k}\right)   \,d\xi
_{k}\left( t\right) ^{\ast }
+\bigg( -iH-\frac{1}{2}\sum_{k}\left(
R_{k}^{\ast }R_{k}-2c_{k}^{\ast }R_{k}+|c_{k}|^{2}\right) \bigg)  dt  
 .
\label{eq:M}
\end{eqnarray}

\section{Representation}
\label{sec:Rep}
For clarity we take the Gisin-Percival-Schr\"{o}dinger equation with just a single collapse operator $R$:
\begin{eqnarray}
d|\tilde{\psi}_{t}\rangle =\left( -iH-\frac{1}{2}\left( R^{\ast
}R-2c^{\ast }R+|c|^{2}\right) \right) |\tilde{\psi}_{t}\rangle \,dt +\left( R-c\right) |\tilde{\psi}_{t}\rangle \,d\xi \left( t\right) ^{\ast
}, 
\label{eq:GS_single_R}
\end{eqnarray}

Our question is whether there exists a choice of $\left( \mathbf{L},H'\right) $ for which the associated Belavkin-Schr\"{o}dinger equation reproduces the 
$\left( R,H\right) $ Gisin-Percival-Schr\"{o}dinger equation.

\begin{proposition}
\label{prop:bv}
There is no direct Belavkin-Schr\"{o}dinger equation which will reproduce the $\left(
R,H\right) $ Gisin-Percival-Schr\"{o}dinger equation, in the sense that 
\begin{eqnarray*}
\vert \psi (t) \rangle \equiv e^{i\theta (t)} \vert \tilde{\psi } (t) \rangle ,
\end{eqnarray*}
for some real valued differentiable process $\theta$.
\end{proposition}

\begin{proof}
Let us suppose that there is a $(\mathbf{L},H')$ Belavkin-Schr\"{o}dinger equation (\ref{eq:BS_equation}) reproducing (\ref{eq:GS_single_R}), up to some ignorable phase term. That is,
\begin{eqnarray*}
dF_{ (\mathbf{L}, H')} (t) =d M_{(R,H )} (t) +i \dot{\theta} (t) \ dt.
\end{eqnarray*}
Looking at the skew-adjoint time-independent terms in the $dt$
coefficients, we see that we need to have the same Hamiltonian, $H'=H$ and $\theta =0$.
Moreover, we then have 
\begin{eqnarray*}
\sum_{k}\text{Re}\left( L_{k}^{\ast }L_{k}-\lambda _{k}\left( t\right) L_{k}+%
\frac{1}{4}\lambda _{k}\left( t\right) ^{2}\right)  \equiv \text{Re}\left(
R^{\ast }R-2c\left( t\right) ^{\ast }R+|c\left( t\right) |^{2}\right) 
\end{eqnarray*}
and comparing the time-independent parts we see that 
\begin{eqnarray}
\sum_{k}L_{k}^{\ast }L_{k}=R^{\ast }R.  \label{eq:cond0}
\end{eqnarray}
In addition, we have 
\begin{eqnarray*}
\sum_{k}\left( L_{k}+L_{k}^{\ast }\right) \lambda _{k}\left( t\right) -\frac{%
1}{2}\sum_{k}\lambda _{k}\left( t\right) ^{2}=2c\left( t\right) ^{\ast
}R+2c\left( t\right) R^{\ast }-2\left| c\left( t\right) \right| ^{2}
\end{eqnarray*}
and averaging gives 
\begin{eqnarray}
\frac{1}{4}\sum_{k}\lambda _{k}\left( t\right)
^{2}=\left| c\left( t\right) \right| ^{2}.
\label{eq:problem1}
\end{eqnarray}
Subtracting this last part off,
we obtain 
\begin{eqnarray}
\sum_{k}\left( L_{k}+L_{k}^{\ast }\right) \lambda _{k}\left( t\right)
=2c\left( t\right) ^{\ast }R+2c\left( t\right) R^{\ast }.
\label{eq:problem2}
\end{eqnarray}
Now, equating the noise terms leads to 
\begin{eqnarray}
\sum_{k}\left( L_{k}-\frac{1}{2}\lambda _{k}\left( t\right) \right)
\,dI_{k}\left( t\right) \equiv \left( R-c\left( t\right) \right) \,d\xi
\left( t\right) ^{\ast }  \label{eq:cond_complex}
\end{eqnarray}
and from the requirement that the $I_{k}$ are independent canonical Wiener
processes and that $\left( d\xi ^{\ast }\right) ^{2}=0$ we get that 
\begin{eqnarray*}
0 =\sum_{k}\left( L_{k}-\frac{1}{2}\lambda _{k}\left( t\right) \right)
^{2} = \sum_{k}(L_{k}^{2}-\lambda _{k}\left( t\right) L_{k}+\frac{1}{4}\lambda
_{k}\left( t\right) ^{2}).
\end{eqnarray*}
As this should be true for time-dependent $\lambda _{k}\left( t\right) $ we
require that 
\begin{eqnarray}
\sum_{k}L_{k}^{2}=0
\label{eq:zero}
\end{eqnarray}
and therefore 
\begin{eqnarray*}
\sum_{k}\lambda _{k}\left( t\right) L_{k}=\frac{1}{4}\sum_{k}\lambda
_{k}\left( t\right) ^{2}.
\end{eqnarray*}
We must then have $\sum_{k}\lambda _{k}\left( t\right) \left(
L_{k}+L_{k}^{\ast }\right) =\frac{1}{2}\sum_{k}\lambda _{k}\left( t\right)
^{2}$. The average of this last equation for vector state $\psi _{t}$ yields 
$\sum_{k}\lambda _{k}\left( t\right) ^{2}=\frac{1}{2}\sum_{k}\lambda
_{k}\left( t\right) ^{2}$ which is a contradiction.
\end{proof}

\subsection{Representation up to a Stochastic Phase}
\label{sec:representation}
We may however try and find a relation such as
\begin{eqnarray}
\vert \psi (t) \rangle \equiv e^{i\Theta (t)} \vert \tilde{\psi } (t) \rangle
\end{eqnarray}
where $\Theta ( t ) ^{\ast }=\Theta ( t ) $ is a self-adjoint stochastic process of unbounded variation.
This now leads to
\begin{eqnarray}
dF_{ (\mathbf{L}, H')} (t) =d M_{(R,H )} (t) +i d\Theta (t) +i d\Theta (t)  \ dM_{(R,H )} (t) .\notag \\
\label{eq:F=M+Theta}
\end{eqnarray}
We may now relax condition (\ref{eq:cond_complex}), and replace it by
\begin{eqnarray}
\sum_{k}\left( L_{k}-\frac{1}{2}\lambda _{k}\left( t\right) \right)
\,dI_{k}\left( t\right)  \equiv \left( R-c\left( t\right) \right) \,d\xi
\left( t\right) ^{\ast }+id\Theta\left( t\right)  . \label{eq:cond_relax}
\end{eqnarray}
This then implies
\begin{eqnarray}
\sum_{k}\left( L_{k}+L_{k}^{\ast }-\lambda _{k}\left( t\right) \right)
\,dI_{k}\left( t\right)  \equiv \left( R-c\left( t\right) \right) \,d\xi
\left( t\right) ^{\ast }+\left( R^{\ast }-c\left( t\right) ^{\ast }\right)
\,d\xi \left( t\right) .  \label{eq:cond_real}
\end{eqnarray}
We will now show that for an elementary class, which we called the Canonical class, it is possible to find
a stochastic phase $\Theta$ so that this is achieved.

For further discussion about the gauge invariance of stochastic master equations when the wavefunction is multiplied by a stochastic phase, see Wiseman and Milburn \cite{WM09}, page 170.
\subsection{The Canonical Class}

Let us suppose that each collapse operator $L_{k}$ is proportional to $R$: 
\begin{eqnarray}
L_{k}\equiv z_{k}R,  \label{eq:cond1}
\end{eqnarray}
then (\ref{eq:cond0}) is satisfied if 
\begin{eqnarray}
\sum_{k}\left| z_{k}\right| ^{2}=1.  \label{eq:cond2}
\end{eqnarray}

We now assume the relaxed condition (\ref{eq:cond_relax}), and from (\ref{eq:cond_real}) we get
\begin{eqnarray*}
\sum_{k}\left( z_{k}R+z_{k}^{\ast }R^{\ast }-\lambda _{k}\left( t\right)
\right) \,dI_{k}\left( t\right) \equiv \left( R-c\left( t\right) \right)
\,d\xi \left( t\right) ^{\ast }+\left( R^{\ast }-c\left( t\right) ^{\ast
}\right) \,d\xi \left( t\right) 
\end{eqnarray*}
where we now have $\lambda _{k}\left( t\right) =z_{k}c\left( t\right)
+z_{k}^{\ast }c\left( t\right) ^{\ast }$. The time-independent coefficients
may be equated and this leads to the identity  $\sum_{k}\left( z_{k}R+z_{k}^{\ast }R^{\ast }\right)
\,dI_{k}\left( t\right) \equiv R\,d\xi \left( t\right) ^{\ast }+R^{\ast
}\,d\xi \left( t\right) $ and assuming that $R\neq R^{\ast }$ one sees that
one should take
\begin{eqnarray}
\xi \left( t\right) ^{\ast }\equiv \sum_{k}z_{k}\,I_{k}\left( t\right) .
\label{eq:can_Noise}
\end{eqnarray}
One easily checks that the rest of (\ref{eq:cond_real}) then follows: that
is, $\sum_{k}\lambda _{k}\left( t\right) \,dI_{k}\left( t\right) \equiv
c\left( t\right) \,d\xi \left( t\right) ^{\ast }+c\left( t\right) ^{\ast
}\,d\xi \left( t\right) $. We also see that (\ref{eq:cond2}) implies that $%
d\xi .d\xi ^{\ast }=dt$, and for the identification (\ref{eq:can_Noise}) we
see that the other condition $\left( d\xi ^{\ast }\right) ^{2}=0$ requires
\begin{eqnarray}
\sum_{k}z_{k}^{2}=0.  \label{eq:cond2a}
\end{eqnarray}
The skew adjoint part of (\ref{eq:cond_relax}) is
\begin{eqnarray*}
\sum_{k}\left( z_{k}R-z_{k}^{\ast }R^{\ast }\right)
\,dI_{k}\left( t\right) \equiv \left( R-c\left( t\right) \right)
\,d\xi \left( t\right) ^{\ast }-\left( R^{\ast }-c\left( t\right)
^{\ast }\right) \,d\xi \left( t\right) +2id\Theta\left( t\right) 
\end{eqnarray*}
and, eliminating terms using (\ref{eq:can_Noise}), we find
\begin{eqnarray*}
d\Theta\left( t\right)  =\frac{1}{2i}\left( c\left( t\right) d\xi \left(
t\right) ^{\ast }-c\left( t\right) ^{\ast }d\xi \left( t\right) \right)  
\equiv \sum_{k}\text{Im}\{z_{k}c(t)\}\,dI_{k}\left( t\right) .
\end{eqnarray*}
We note that for the canonical class the phase $\Theta$ is proportional to the identity operator on $\mathfrak{h}$,
so we may think of it as a (stochastic) phase function.

This stochastic ``phase'' $\Theta$ has the nontrivial Ito table
\begin{eqnarray*}
d\Theta \, d\xi ^{\ast } &=&d\xi ^{\ast }\, d\Theta=-\frac{1}{2i}c\left( t\right) dt, \\
d\Theta \, d\xi  &=&d\xi \, d\Theta=\frac{1}{2i}c\left( t\right) ^{\ast }dt, \\
d\Theta \, d\Theta &=&\frac{1}{2}\left| c\left( t\right) \right| ^{2}dt.
\end{eqnarray*}

Now let us check that this gives the correct answer, specifically, that we obtain the correct $dt$ terms.
Starting from (\ref{eq:F}) for the canonical class, we note (from the conditions $\sum_{k} |z_{k}|^{2}=1, \sum_{k}z_{k}^{2}=0$) that
\begin{eqnarray*}
\sum_k L^\ast_k L_k &=& R^\ast R, \notag \\
\sum_k L^\ast_k \lambda_k &=& \sum_k (z_k c +z_k^\ast c^\ast ) z_k R = c^\ast R ,\notag \\
\sum_k  \lambda_k^2 &=& \sum_k (z_k c +z_k^\ast c^\ast )^2 = 2 |c|^2, \notag \\
\sum_k (L_k - \frac{1}{2} \lambda_k) I_k  &=& \sum_k (z_k R - \frac{1}{2} c z_k - \frac{1}{2} +z_k^\ast c^\ast ) I_k 
=
(R- \frac{1}{2} c ) d \xi^\ast - \frac{1}{2} c^\ast s \xi
.
\end{eqnarray*}
So for the canonical class we have
\begin{eqnarray*}
dF_{ (\mathbf{L}, H)} = \big( -iH - \frac{1}{2}(R^\ast R \underline{- c^\ast R + \frac{1}{2} |c|^2 }) \big) \, dt  + (R- \frac{1}{2} c) \, d \xi^\ast  - \frac{1}{2} c^\ast \, d\xi.
\end{eqnarray*}
The underlined terms are half what they should be in the $dt$ part of $dM_{(R,H)}$ - this was exactly the problem we ran into in Proposition
\ref{prop:bv}. However, let us look at the effect of the stochastic phase $\Theta$. The right hand side of (\ref{eq:F=M+Theta}) is
\begin{eqnarray*}
d M_{(R,H )} (t) +i d\Theta (t) +i d\Theta (t)  \ dM_{(R,H )} (t)   &=&    \bigg( -iH - \frac{1}{2}(R^\ast R  - 2c^\ast R +   |c|^2 ) \bigg) \, dt \\
&&+(R-c) \, d \xi^\ast  +  \frac{1}{2} \big( c \, d\xi^\ast - c^\ast \, d\xi \big)    +  \frac{1}{2} \big( c \, d\xi^\ast - c^\ast \, d\xi \big)(R-c) \, d \xi^\ast   ,
\end{eqnarray*}
and the Ito correction $i d\Theta (t) dM_{(R,H )} (t) \equiv - \frac{1}{2}   c^\ast    (R-c) \, dt$ gives precisely the missing
$dt$ term contribution. By inspection, we see that we have recovered (\ref{eq:F=M+Theta}).

\begin{remark}
For $n=2$, the equations $\sum_k \vert z_k \vert^2 =1$ and $\sum_k z_k^2 =0$ have solution
\begin{eqnarray}
z_1 = \frac{1}{\sqrt{2}}e^{i \phi} , \quad z_1 = \pm i \frac{1}{\sqrt{2}}e^{i \phi},
\end{eqnarray}
for some phase $\phi$.
For each integer $n \ge 3$, we have a larger set of solutions, but a subclass is given by $\{ z_1 , \cdots , z_n \}$ where $z_k = \frac{1}{\sqrt{n}} e^{i\pi (k -1)/n}$, for $k =1, \cdots , n$.
\end{remark}

\subsection{Simplest Example}
\label{subsec:example}
From the remark, we see that the simplest realization (up to a phase) of the canonical class is given by taking $n=2$
collapse operators 
\begin{eqnarray}
L_{1}=\frac{1}{\sqrt{2}}R,\quad L_{2}=\frac{i}{\sqrt{2}}R.  \label{eq:Can_L}
\end{eqnarray}
Here we will have 
\begin{eqnarray*}
\lambda _{1}\left( t\right) &=&\frac{1}{\sqrt{2}}\langle \psi _{t}|\left(
R+R^{\ast }\right) \psi _{t}\rangle , \\
\lambda _{2}\left( t\right) &=&\frac{i}{\sqrt{2}}\langle \psi _{t}|\left(
R-R^{\ast }\right) \psi _{t}\rangle ,
\end{eqnarray*}
so that 
\begin{eqnarray*}
\langle \psi _{t}|R\psi _{t}\rangle =\frac{1}{\sqrt{2}}\left( \lambda
_{1}\left( t\right) -i\lambda _{2}\left( t\right) \right) ,
\end{eqnarray*}
and the complex Wiener process is 
\begin{eqnarray}
\xi \left( t\right) ^{\ast }=\frac{1}{\sqrt{2}}I_{1}\left( t\right) +\frac{i%
}{\sqrt{2}}I_{2}\left( t\right) .  \label{eq:Can_Xi}
\end{eqnarray}

\subsection{The Filters}

Let $X$ be an arbitrary system operator, then its filtered expectation at time $t$ from the Belavkin theory is 
\begin{eqnarray}
\pi _{t}\left( X\right) =\langle \psi _{t}|X\,\psi _{t}\rangle
\label{eq:B_estimate}
\end{eqnarray}
and from (\ref{eq:BS_equation}) we get 
\begin{eqnarray}
d\pi _{t}\left( X\right) =\pi _{t}\left( \mathcal{L}_{\left( \mathbf{L}%
,H\right) }X\right) dt +\sum_{k}\left\{ \pi _{t}\left( XL_{k}+L_{k}^{\ast
}X\right) -\lambda _{k}\left( t\right) \pi _{t}\left( X\right) \right\}
dI_{k}\left( t\right) ,  \label{eq:B_filter}
\end{eqnarray}
and we recall that $\lambda _{k}\left( t\right) $  now equals $\pi _{t}\left( L_{k}+L_{k}^{\ast }\right) $. Here the Lindbladian is the one determined by collapse operators $\mathbf{L}=\left\{ L_{k}\right\} $ and Hamiltonian $H$.

Although not interpreted as a filter, we may consider the equivalent in Gisin-Percival's theory 
\begin{eqnarray}
\tilde{\pi}_{t}\left( X\right) =\langle \tilde{\psi}_{t}|X\,\tilde{\psi}%
_{t}\rangle .  \label{eq:G_estimate}
\end{eqnarray}
This time, using (\ref{eq:GS_equation}) we find 
\begin{eqnarray}
d\tilde{\pi}_{t}\left( X\right) &=&\tilde{\pi}_{t}\left( \mathcal{L}_{\left( 
\mathbf{R},H\right) }X\right) dt +\sum_{k}\left\{ \tilde{\pi}_{t}\left(
XR_{k}\right) -c_{k}\left( t\right) \tilde{\pi}_{t}\left( X\right) \right\}
\,d\xi _{k}\left( t\right) ^{\ast }  \notag \\
&&+\sum_{k}\left\{ \tilde{\pi}_{t}\left( R_{k}^{\ast }X\right) -c_{k}\left(
t\right) ^{\ast }\tilde{\pi}_{t}\left( X\right) \right\} \,d\xi _{k}\left(
t\right) . \notag \\
 \label{eq:G_filter}
\end{eqnarray}

As the Belavkin-Schr\"{o}dinger wave function $| \psi _t \rangle$ in the canonical class is equal to the Gisin-Percival-Schr\"{o}dinger wavefunction up to a phase, the following result should not be surprising.

\begin{proposition}
\label{prop:bv2}
The Belavkin filter $\pi _{t}\left( X\right) $ corresponding to the canonical class $\left( L_{1}=\frac{1}{\sqrt{2}}R,L_{2}=\frac{i}{\sqrt{2}}R,H\right) $ Belavkin-Schr\"{o}dinger model is identical to the Gisin-Percival ``filter'' $\tilde{\pi}_{t}\left( X\right) $
for the\thinspace $\left( R,H\right) $ Gisin-Percival-Schr\"{o}dinger equation.
\end{proposition}

The proof is routine and simply amounts to substituting (\ref{eq:Can_L}) into (\ref{eq:B_filter}) and reassembling the components using the definition (\ref{eq:Can_Xi}) for the complex Wiener process. The result is just the single collapse operator $R$ version of (\ref{eq:GS_equation}). If both are initialized on the same state $\psi _{0}$ then we must have $\pi _{t}\equiv \tilde{\pi}_{t}$. The extension to the multi-dimensional situation is obvious.

\section{Systems Theory Approach}
\label{sec:system}

We recall the ``SLH'' formulation of open quantum Markov systems. Our
system, with Hilbert space $\mathfrak{h}$, is coupled to an environment which is
a Bose reservoir with Fock space $\mathfrak{F}$. For an $n$ input model, we take 
$\mathfrak{F}$ to be the Bose Fock space with one particle space $\mathbb{C}%
^{n}\otimes L^{2}[0,\infty )$. Let $\left\{ e_{1},\cdots ,e_{n}\right\} $ be
an orthonormal basis for $\mathbb{C}^{n}$, then $e_{k}\otimes f$ gives a one
particle state corresponding to an input quantum in the $k$th channel with
wave function $f=f\left( t\right) $, $t\geq 0$. (We may think of the
reservoir quanta traveling through the system, and the parameter $t$ labels
the part of the input field passing through the system at time $t$). We take 
$B_{k}\left( t\right) $ to be the annihilation operator for the one-particle
vector $e_{k}\otimes 1_{\left[ 0,t\right] }$ - that is, annihilating a
reservoir quantum of type $k$ sometime over the interval $\left[ 0,t\right] $%
. Along with the creators, we have the canonical commutation relations
\begin{equation*}
\left[ B_{j}\left( t\right) ,B_{k}\left( s\right) ^{\ast }\right] =\delta
_{jk}\,\text{max}\left\{ t,s\right\} .
\end{equation*}
Formally, we may introduce densities $b_{k}\left( t\right) $ interpreted as
annihilators of a reservoir quantum of type $k$ at time $t$, with $\left[
b_{j}\left( t\right) ,b_{k}\left( s\right) ^{\ast }\right] =\delta
_{jk}\delta \left( t-s\right) $. Then $B_{j}\left( t\right) \equiv
\int_{0}^{t}b_{j}\left( s\right) ds$, etc. We may furthermore introduce the
operators
\begin{equation*}
\Lambda _{jk}\left( t\right) \equiv \int_{0}^{t}b_{j}\left( s\right) ^{\ast
}b_{k}\left( s\right) ds
\end{equation*}
describing the instantaneous scattering from channel $k$ to channel $j$ of
the reservoir quanta at some time $s\in \left[ 0,t\right] $.

It is well known that the most general form of a unitary adapted process on $%
\mathfrak{h}\otimes \mathfrak{F}$ from a quantum stochastic differential equation
with constant coefficients is $U\left( t\right) $ given by
\begin{equation*}
dU\left( t\right) =dG\left( t\right) \,U\left( t\right) ,\quad U\left(
0\right) =I,
\end{equation*}
where
\begin{eqnarray}
dG\left( t\right)  =-\left( iH+\frac{1}{2}\sum_{k}L_{k}^{\ast
}L_{k}\right) \otimes dt  +\sum_{k}L_{k}\otimes dB_{k}\left( t\right) ^{\ast } 
-\sum_{jk}L_{j}^{\ast }S_{jk}\otimes dB_{k}\left(
t\right)
+\sum_{jk}\left(
S_{jk}-\delta _{jk}\right) \otimes d\Lambda _{jk}\left( t\right) 
\end{eqnarray}
where the coefficients $S_{jk},L_{j},H$ are operators on $\mathfrak{h}$ with $S=%
\left[ S_{jk}\right] $ unitary and $H=H^{\ast }$. We write $G\sim \left( S,%
\mathbf{L},H\right) $ for the coefficients, and refer to them as the
Hudson-Parthasarathy parameters.

We set $j_{t}\left( X\right) =U\left( t\right) ^{\ast }\left[ X\otimes I%
\right] U\left( t\right) $ for a system operator $X$, and we have
\begin{eqnarray*}
dj_{t}\left( X\right)  &=&j_{t}\left( \mathcal{L}X\right) \otimes
dt +\sum_{jk}j_{t}\left( \left[ L_{j}^{\ast },X\right] S_{jk}\right) \otimes
dB_{k}\left( t\right)  \\
&&+\sum_{jk}j_{t}\left( S_{jk}^{\ast }\left[ X,L_{k}\right] \right) \otimes
dB_{j}\left( t\right) ^{\ast } +\sum_{jk}\left( \sum_{l}S_{lj}^{\ast }XS_{lk}-\delta _{jk}X\right)
\otimes d\Lambda _{jk}\left( t\right) 
\end{eqnarray*}
where $\mathcal{L}X=\frac{1}{2}\sum_{k}\left[
L_{k}^{\ast },X\right] L_{k}+\frac{1}{2}\sum_{k}L_{k}^{\ast }\left[ X,L_{k}%
\right] -i\left[ X,H\right] $ is a GKS-Lindblad generator.

The output fields are
\begin{equation*}
B_{\text{out},k}\left( t\right) =U\left( t\right) ^{\ast }\left[ I\otimes
B_{k}\left( t\right) \right] U\left( t\right) 
\end{equation*}
and so $dB_{\text{out},j}\left( t\right) =\sum_{k}j_{j}\left(
S_{jk}\right) \otimes dB_{k}\left( t\right) +j_{t}\left( L_{k}\right)
\otimes dt.$

\subsection{Filtering}

Suppose we wish to monitor the output quadratures
\begin{equation*}
Y_{k}\left( t\right) =B_{\text{out},k}\left( t\right) +B_{\text{out}%
,k}\left( t\right) ^{\ast }.
\end{equation*}
They form a commuting set of observables. We have that $Y_{k}\left( t\right)
\equiv U\left( t\right) ^{\ast }\left[ I\otimes Z_{k}\left( t\right) \right]
U\left( t\right) $ where $Z_{k}\left( t\right) =B_{k}\left( t\right)
+B_{k}\left( t\right) ^{\ast }$: they are a commuting set of observables on
the Fock space, having the distribution of independent Wiener processes for
the Fock vacuum state.

The aim of filtering is to compute
\begin{equation*}
\pi _{t}\left( X\right) =\mathbb{E}\left[ j_{t}\left( X\right) |\mathfrak{Y}_{t}%
\right] 
\end{equation*}
which is the conditional expectation of $j_{t}\left( X\right) $ onto the
(commutative) algebra $\mathfrak{Y}_{t}$ generated by the measured observables $%
\left\{ Y_{k}\left( s\right) :k,0\geq s\leq s\right\} $. 

We restrict our attention to the non-scattering case $S=I_{n}$, that is $%
G\sim \left( I_{n},\mathbf{L},H\right) $ and
\begin{eqnarray*}
dG\left( t\right)  =-\left( iH+\frac{1}{2}\sum_{k}L_{k}^{\ast
}L_{k}\right) \otimes dt -\sum_{k}L_{k}^{\ast }\otimes dB_{k}\left( t\right) +\sum_{k}L_{k}\otimes
dB_{k}\left( t\right) ^{\ast },
\end{eqnarray*}
so that
\begin{equation*}
dY_{k}\left( t\right) =I\otimes dZ_{k}\left( t\right) +j_{t}\left(
L_{k}+L_{k}^{\ast }\right) \otimes dt.
\end{equation*}
In this case, the filter is then given by
\begin{equation*}
\pi _{t}\left( X\right) =\langle \psi _{t}|X\,\psi _{t}\rangle 
\end{equation*}
where $\psi _{t}$ is the solution to the Belavkin-Schr\"{o}dinger equation
(\ref{eq:BS_equation}). The processes $I_{k}\left( t\right) $ are the innovations:
\begin{eqnarray*}
dI_{k}\left( t\right)  =dY_{k}-\pi _{t}\left( L_{k}+L_{k}^{\ast }\right) dt
=I\otimes dZ_{k}\left( t\right)  +\left[ j_{t}\left( L_{k}+L_{k}^{\ast }\right) -\pi _{t}\left(
L_{k}+L_{k}^{\ast }\right) \right] \otimes dt.
\end{eqnarray*}
It is easy to see that the innovations form a multidimensional Wiener
process.

\subsection{The Series Product}

We now consider the situation where the output of one system, $G_{1}\sim
\left( S_{1},\mathbf{L}_{1},H_{1}\right) $ is fed in as input to another, $%
G_{2}\sim \left( S_{2},\mathbf{L}_{2},H_{2}\right) $. In the limit of
instantaneous feedforward, we find the combined model $G_2 \vartriangleleft G_1$, where the series product is defined by 
\begin{eqnarray*}
\left( S_{2},\mathbf{L}_{2},H_{2}\right) \vartriangleleft \left( S_{1},%
\mathbf{L}_{1},H_{1}\right) =\left( S_{2}S_{1},\mathbf{L}_{2}+S_{2}\mathbf{L}%
_{1},H_{1}+H_{2}+\text{Im}\left\{ \mathbf{L}_{2}^{\ast }S_{2}\mathbf{L}%
_{1}\right\} \right) .
\end{eqnarray*}

The covariance of the Lindblad generator, discussed in Subsection \ref{subec:Cov}, can be described in the following terms.
A triple $\left( U,\mathbf{\beta },\varepsilon \right) $ is said to belong to the central extension of the Euclidean group over the Hilbert space, denoted by $Eu\left( \mathfrak{h}\right) $ if $U$ is an $n\times n$ matrix with complex scalar entries, $\beta =\left[ \beta _{1},\cdots ,\beta _{n}\right] ^{\top }$ is column vector of complex scalars, and $\varepsilon $ is real. We then have the covariance 
\begin{eqnarray*}
\mathcal{L}_{E\vartriangleleft G}=\mathcal{L}_{G}
\end{eqnarray*}
for all $E\in Eu\left( \mathfrak{h}\right) $. 
Note that (\ref{eq:EU_rot}) and (\ref{eq:EU_transl}) correspond to $(I_n,\mathbf{L},H ) \mapsto (U,0,0) \vartriangleleft (I_n,\mathbf{L},H ) $ and
$(I_n,\mathbf{L},H ) \mapsto (I_n,\mathbf{\beta},\epsilon) \vartriangleleft (I_n,\mathbf{L},H ) $, respectively.

\subsection{Weyl Displacement}
Let $\mathbf{\beta}$ be a collection of square-integrable complex-valued function $\beta_k = \beta_k (t)$. We can consider a \lq\lq Weyl Box\rq\rq \, to be a component which displaces the vacuum inputs by the amplitudes $\mathbf{\beta}(t)$. That is,
\begin{eqnarray}
\text{Weyl}_{\mathbf{\beta}} (t) \sim ( I_n , \mathbf{\beta(t) } , 0).
\end{eqnarray}
The output annihilator process of the Weyl Box will then be $B_k (t) + \int_0^t \beta_k (s) ds$.

Placing a Weyl Box after a system $G \sim (I_n , \mathbf{L} , H)$ as a post-filter results in the combined model
\begin{eqnarray*}
  \text{Weyl}_{\mathbf{\beta}} (t)\vartriangleleft G \sim \bigg( I_n , \mathbf{L+\beta(t) } , H+ \mathrm{Im}
	\sum_k  \beta^\ast_k (t)L_k \bigg).
\end{eqnarray*}
Note that the Hudson-Parthasarathy coefficients are now time-dependent.

\section{Simulating the Gisin-Percival equation}
\label{sec:simulation}
We now show that it is possible to use our system theory approach to build a feedback systems such that the conditioned state corresponds to the Gisin-Percival Schr\"{o}dinger wavefunction.
The set-up is described in Figure \ref{fig:GP_filter} below. We take the system $G$ to be the simple model $\left( L_{1}=\frac{1}{\sqrt{2}}R,L_{2}=\frac{i}{\sqrt{2}}R,H\right) $, that is,
\begin{eqnarray}
G \sim \bigg( I_2 , 
\left[
\begin{array} {c}
	\frac{1}{\sqrt{2}}R \\
		\frac{i}{\sqrt{2}}R
\end{array}
\right] , H \bigg).
\end{eqnarray}
Both inputs are displaced by amplitudes $\beta_k (t)$ using Weyl Box post-filters lead to the composite system 
$ G  \vartriangleleft \text{Weyl}_{\mathbf{\beta}} (t)$ given by
\begin{eqnarray}
G &\sim& \bigg( I_2 , 
\left[
\begin{array} {c}
	\frac{1}{\sqrt{2}}R + \beta_1 (t) \\
		\frac{i}{\sqrt{2}}R + \beta_2 (t)
\end{array}
\right] ,  H + \frac{1}{\sqrt{2}} \text{Im} \bigg\{  [ \beta_1 (t) +i \beta_2 (t) ]^\ast R \bigg\} \bigg).
\label{eq:displaced}
\end{eqnarray}

We perform homodyne measurements on the quadratures, recording the essentially classical processes $Y_1 (t)$ and $Y_2(t)$. Using this, we may compute the conditional wave function for the Belavkin-Schr\"{o}dinger equation, but now with the parameters $(I_2,\mathbf{L}, H)$ replaced by the modulated ones in (\ref{eq:displaced}).

\begin{figure}[htbp]
	\centering
		\includegraphics[width=0.4750\textwidth]{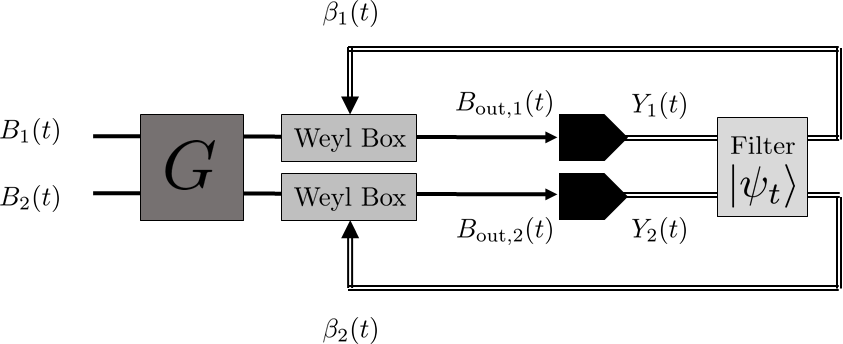}
	\caption{The set-up is as follows. The overall system is driven by a pair of vacuum quantum input fields $B_1(t)$ and $B_2 (t)$.
	These drive the system and the output is subsequently modulated by Weyl Boxes with displacements $\beta_1 (t)$ and $\beta_2 (t)$ respectively. We perform a pair of homodyne measurements on the (modulated) quadratures $Y_k (t) =B_{\text{out},k}(t) +B_{\text{out},k}(t)^\ast$ and use this information to compute the conditioned state $\vert \psi_t \rangle$ using the Belavkin filter. Using the filter we compute a pair of processes $\alpha_k (t)$ which are then fed back in as the displacements ($\beta_k (t)$) for the respective Weyl Boxes.}
	\label{fig:GP_filter}
\end{figure}

At this stage we compute a pair of processes $\alpha_k (t)$ using the conditional state $\psi_t$. For the moment we make only one assumption:

(A) \textit{The functions $\alpha_k (t)$ are purely imaginary functions, depending causally on the measurement records $\{ Y_k (s): k=1,2; 0 \le s \le t \}$.}

The input displacements are then taken to be $\beta_k (t)$ which are Fock space based processes defined by the \emph{pull-back transformation}
\begin{eqnarray}
I \otimes \beta_k (t) = U(t) \alpha_k (t) U(t)^\ast.
\end{eqnarray}
Note that the $\beta_k (t)$ defined in this way are functions of the measurement records $\{ Z_k (s): k=1,2; 0 \le s \le t \}$.

We now try and introduce feedback following the approach in \cite{Gough_JMP_2017}. 
In terms of consistency, we are now taking the displacements $\beta_k (t)$ to be adapted processes on the Fock space, and describable in terms of the processes $Z_1$ and $Z_2$.
By this manner we makes the feedback loops. While the model is more involved than the constant coefficient, or even time-dependent coefficient case, the dynamics is still well defined, and
arguably markovian.

Let us look at the resultant filter in detail. First we have, $dY_k (t) = dZ_k (t) + j_t (L_k+ \beta_k (t) + \text{H.c.}) \, dt$. But as the $\beta_k$ are imaginary, this reduces to
\begin{eqnarray}
dY_1 (t) &=& dZ_1 (t) + \frac{1}{2 \sqrt{2}} j_t (R+R^\ast ) \, dt, \notag \\
dY_2 (t) &=& dZ_2 (t) + \frac{i}{2 \sqrt{2}} j_t (R-R^\ast ) \, dt ,
\end{eqnarray}
and in both cases $dI_k (t) \equiv dY_k (t) - \pi_t (L_k + L^\ast_k ) \, dt$.

The Belavkin-Schr\"{o}dinger equation is now the one with the term (\ref{eq:F}) given by
\begin{eqnarray}
dF (t)  &=&\sum_{k}\big( L_{k} + \alpha_k (t) -\frac{1}{2}\lambda _{k}\big)  
\,dI_{k}\left( t\right) +\bigg( -iH-i \sum_k \text{Im}\{ \alpha^\ast_k (t) L_k \}-\frac{1}{2}\sum_{k}\big( (L_{k}+\alpha_k (t))^{\ast
}\big( L_{k}+\alpha_k (t) \big) \notag \\
&&+\frac{1}{2}\sum_k \lambda _{k}( L_{k}+\alpha_k (t) ) -\frac{1}{8}\sum_k \lambda _{k}^{2}\big) \bigg)   dt 
 .
\label{eq:mod_F}
\end{eqnarray}
Some comments are in order. First of all (\ref{eq:mod_F}) is obtained from (\ref{eq:F}) with the change
$L_k \mapsto L_k + \alpha_k (t)$ and $H \mapsto H + \sum_k \text{Im}\{ L^\ast_k \alpha_k (t) \}$. We note that these equations involve the $\alpha_k (t)$,
which are function of the measurements $\{ Y_k \}$, rather than the $\beta_k (t)$ and the reason for this is that we have to push forward to the 
output picture, see \cite{BvH_ref}.
The $\lambda_k(t)$ should in principle be shifted to $\lambda_k (t) + \alpha_k (t) + \alpha_k (t)^\ast$. But this shift is by the real part of the $\beta_k (t)$
and we recall that these vanish by assumption (A).

We now seek to arrange things so that (\ref{eq:mod_F}) gives us the Gisin-Percival term $dM_{(R,H)} (t)$ and for the noise term, with the previous identification (\ref{eq:Can_Xi}) for the noise, we see that we must have
\begin{eqnarray}
L_1 + \alpha_1 (t) -\frac{1}{2}\lambda _1 (t) &\equiv& \frac{1}{\sqrt{2}} \big( R - c (t) \big), \notag \\
L_2 + \alpha_2 (t) -\frac{1}{2}\lambda _2 (t)  &\equiv& \frac{i}{\sqrt{2}} \big( R - c (t) \big),
\end{eqnarray}
and this implies the following choice
\begin{eqnarray}
\alpha_1 (t) & \equiv & - \frac{1}{2 \sqrt{2}} \langle \psi_t \vert ( R-R^\ast ) \psi_t \rangle , \notag \\
\alpha_2 (t) & \equiv & -\frac{i}{2 \sqrt{2}} \langle \psi_t \vert ( R+R^\ast ) \psi_t \rangle .
\label{eq:alpha}
\end{eqnarray}
We remark that both of these choices determine purely imaginary expressions, and so the ansatz (A) is in place. In fact, if we decompose $c(t)=\langle \psi_t \vert R\, \psi_t \rangle$ into real and imaginary parts $c'(t) + i c''(t)$, then
\begin{eqnarray}
\lambda_1 (t) = \sqrt{2} c'(t), \quad \lambda_2 (t) = - \sqrt{2} c''(t) \notag \\
\alpha_1 (t) = - \frac{i}{\sqrt{2}} c''(t), \quad \alpha_2 (t) =- \frac{i}{\sqrt{2}} c'(t).
\end{eqnarray}

With this choice of the $\alpha_k (t)$'s, we find that the $dt$ term in (\ref{eq:mod_F}) is
\begin{eqnarray}
 -iH -\frac{1}{2}\sum_{k} L_{k}^\ast L_k - \sum_k \big( \alpha_k^\ast (t) -\frac{1}{2}\lambda_k (t) \big) L_k  
 -\frac{1}{2}\sum_{k}\big(  \alpha_k (t))^\ast  \alpha_k (t)  
-\lambda _{k}\alpha_k (t)  +\frac{1}{4}\lambda _{k}^{2}\big)  
 .
\end{eqnarray}
and it is straightforward to check that
\begin{eqnarray}
&&\sum_k \big( \alpha_k^\ast (t) -\frac{1}{2}\lambda_k (t) \big) L_k \equiv c(t)^\ast R, \notag \\
&&\frac{1}{2}\sum_{k}   \alpha_k (t))^\ast  \alpha_k (t)  = \frac{1}{4}\sum_k \lambda _{k}^{2} \equiv \frac{1}{2}|c(t)|^2, \notag \\
&& \sum_k \lambda _{k}\alpha_k (t)  =0.
\end{eqnarray}
Substituting these explicit expressions in to (\ref{eq:mod_F}) leads to $dF(t) = dM_{(R,H)} (t)$, which is now exactly
as in (\ref{eq:M}). In other words, the set-up with feedback leads to a conditioned state dynamics which is the same as the Gisin-Percival-Schr\"{o}dinger equation.

\section{Conclusion}
The Gisin-Percival equation appears distinct from the usual stochastic Schr\"{o}dinger equations obtained through quantum filtering models. However, as shown by Wiseman and Milburn \cite{WM93} it may be obtained as a high detuning limit of a homodyne measurement. Here we have shown that the Gisin-Percival equation is equivalent up to an overall phase to a specific quantum trajectories problem (in fact, we obtain a general class, and work with the simplest such representative). The phase term will however be a stochastic process with non-trivial quadratic variation. The quantum measurement scheme needs to ensure that the covariance symmetry of the Gisin-Percival equation hold, and in the simplest case this is done by indirect measurement where two input-output channels couple to the same system operator $R$ and we make homodyne measurements of orthogonal quadratures of the output fields.

In fact, we show that an appropriate feedback of the measured signals allows us to include this phase term, leading to an exact analogue simulation of the Gisin-Percival state diffusion in terms of quantum trajectories.

The quantum state diffusion equation of Gisin and Percival is therefore not distinguishable from standard quantum trajectory models.

\bigskip

\begin{acknowledgments}
The author wishes to thank Hendra Nurdin for several useful discussions, and for pointing out preprint \cite{Partha_Usha}, as well as several discussions with K.R. Parthasarathy and A.R. Usha Devi.
\end{acknowledgments}

\quad

\end{document}